\documentclass[12pt]{article}
\usepackage[utf8]{inputenc}
\usepackage{amsfonts}
\usepackage{float}
\usepackage{amsmath,algorithm}
\usepackage{fullpage}
\usepackage{graphicx}
\usepackage{url}
\usepackage{amsthm}
\usepackage{braket}
\usepackage{extarrows}
\usepackage{titling}

\theoremstyle{definition} 
\theoremstyle{definition} 
\newtheorem {theorem} {Theorem}

\newtheorem {lemma} {Lemma}

\newtheorem{protocol}{Protocol}

\floatname{algorithm}{Protocol}

\DeclareMathOperator{\Tr}{Tr}

\def\ds{\displaystyle}

\newcommand{\walkop}{U_{\text{walk}}}
\begin{document}
\author{
C. Vlachou$^{1,2,6}$, 
W. Krawec$^{3}$, 
P. Mateus$^{1,2}$,\\ 
N. Paunkovi\'c$^{1,2,6}$ and 
A. Souto$^{1,4,5}$
\\ \ \\
${}^1$ SQIG -- Instituto de Telecomunica\c c\~oes\\
${}^2$ Dep. de Matem\'atica -- Instituto Superior T\'ecnico, Universidade de Lisboa\\
${}^3$ Computer Science and Engineering Department, \\University of Connecticut, Storrs, CT 06268 USA\\
${}^4$ LaSige -- Faculdade de Ci\^encias, Universidade de Lisboa\\
${}^5$ Dep. de Inform\'atica -- Faculdade de Ci\^encias, Universidade de Lisboa\\
${}^6$ CeFEMA, Instituto Superior T\'ecnico, Universidade de Lisboa}
\title{Quantum key distribution with quantum walks}

\maketitle
\begin{abstract}
Quantum key distribution is one of the most fundamental cryptographic protocols. Quantum walks are important primitives for computing. In this paper we take advantage of the properties of quantum walks to design new secure quantum key distribution schemes. In particular, we introduce a secure quantum key-distribution protocol equipped with verification procedures against full man-in-the-middle attacks. Furthermore, we present a one-way protocol and prove its security. Finally, we propose a semi-quantum variation and prove its robustness against eavesdropping.\end{abstract}

\section{Introduction}

Quantum Key Distribution (QKD) is the most secure and practical instance of quantum cryptography. We recall that a key distribution scheme is a protocol between two parties with the purpose of sharing a common string (the key), which afterwards, they can use to communicate privately, in a pre-agreed encryption scheme. Therefore, it is required that any third party that might be eavesdropping is not able to extract information about the key, and thus compromizing the privacy of the communication. Bennet and Brassard~\cite{ben:bra:84} in 1984, and Ekert~\cite{eke:91} in 1991, proposed the first QKD protocols, upon which all QKD protocols are based. Since then a lot of modifications and improvements have been proposed in order to achieve unconditionally secure and practical QKD schemes, by taking advantage of the physical laws of quantum mechanics. For a review, see~\cite{sca:bec:cer:dui:dus:lut:pee:09}. Several QKD experiments over long distances have been reported~\cite{urs:etal:07,tak:sas:tam:koa:15,pug:kai:bou:jin:sul:agn:ani:mak:cho:hig:jen:17,idquantique}, and QKD is already commercial.\footnote{Currently there are three companies offering commercial QKD systems: ID Quantique (Geneva), MagiQ Technologies, Inc. (New York) and QuintessenceLabs (Australia).} Furthermore, the recent successful launch of a satellite~\cite{yin:etal:17} paved the way for intercontinental QKD.

Quantum walks (QW) have been introduced in 1993, in~\cite{aha:dav:zag:93}, as the quantum analogue of classical random walks. Since then, they have been playing a major role in quantum computing theory, as their applications vary from quantum algorithms~\cite{far:gut:98,chi:etal:03,amb:03,san:08,por:13} to universal quantum computing schemes~\cite{chi:09,lov:coo:eve:tre:ken:10,chi:gos:web:13}. 

Recently, the application of QWs to the creation of actual quantum cryptographic protocols has been investigated.  For instance, in \cite{roh:fit:gil:15}, Rohde {\em et al.}, proposed a limited form of quantum homomorphic encryption using multi-particle QWs. In their protocol, a server could manipulate data sent by a client in such a way that, first, the server has limited information on the client's data while, second, the client has limited information on the server's computation.

In this paper, we revisit the public-key cryptosystem~\cite{vla:rod:mat:pau:sou:15}, which is based on QWs, in order to construct secure QKD protocols. First, we suitably modify~\cite{vla:rod:mat:pau:sou:15}, so that the quantum state generated by means of a QW encodes the secret key as opposed to the message; such a key could be used later as input to a one-time-pad encryption system gaining information theoretic security for message delivery. Our motivation is that QKD schemes have several advantages, which we present in due course, over public-key cryptosystems. The modification of the original public-key system is non-trivial, however, and requires care as we can no longer rely on the existence of a trusted mechanism for public-key delivery (such as a public-key infrastructure), as is typically assumed in quantum public-key cryptography \cite{nik:08,sey:nik:alb:12,vla:rod:mat:pau:sou:15}. 

While the above QKD protocol is two-way, i.e., both Alice ($A$) and Bob ($B$) perform QW operations, we also construct a one-way QKD protocol, where again the key is encoded in a QW state. In this case, it is only Alice that chooses randomly the precise QW to encode the key, while Bob is randomly choosing in which basis (computational or QW) to measure in order to obtain it. After disclosing their choices by means of classical communication, they are able to establish a shared key. We prove that the protocol is secure against general attacks, even if the eavesdropper Eve $(E)$ has great advantage over Alice and Bob. 

As a third contribution in this paper, we propose a new \emph{semi-quantum} key-distribution (SQKD) protocol based on QWs.  Semi-quantum cryptography was first introduced in 2007 by Boyer {\em et al.}, in \cite{boy:ken:mor:07,boy:gel:ken:mor:09} as a way to study ``how quantum'' does a protocol need to be in order to gain an advantage over its classical counterpart -- namely, how quantum do the parties need to be in order to establish a secret key secure against an all powerful adversary.  Using classical communication alone, this task is impossible -- indeed, any key distribution protocol, relying only on classical communication, cannot be unconditionally secure and, instead, requires computational hardness assumptions to be made on the adversary. On the other hand, QKD protocols do have provable unconditional security. A semi-quantum protocol places severe restrictions on one of the participating users (typically Bob) in that he may only operate in a ``classical'' or ``semi-quantum'' manner. Namely, this limited user can only directly work with the computational $Z$ basis.  No restrictions are placed on the other participant Alice, and of course, no restrictions are placed on Eve.

The paper is organized as follows: In Section~\ref{sec:perliminareis}, we provide a brief introduction to QWs, with all the information and notation used throughout the paper. In Section~\ref{sec:qkdscheme}, we present a secure two way QKD scheme based on QWs, which is a modification of the public-key cryptosystem in Ref.~\cite{vla:rod:mat:pau:sou:15}. We give motivation for this modification and furthermore, we propose two different verification procedures against full man-in-the-middle attacks. In Section~\ref{sec:oneway}, we introduce a one-way QKD protocol, which we prove to be unconditionally secure, by reducing it to an equivalent entanglement-based protocol. We provide our numerical results for the optimal choice of the QW parameters that maximize the noise tolerance of the protocol. Finally, in Section~\ref{sec:semi-qkdscheme} we provide a SQKD protocol and we show its robustness against eavesdropping.

\section{Quantum Walk Preliminaries}
\label{sec:perliminareis}
In this paper we consider QWs on a circle. In this case, the walker hops along discrete positions on a circle. The Hilbert space $\mathcal{H}$, describing the QW, is the tensor product of the positions Hilbert space $\mathcal{H}_p$ 
and the coin Hilbert space $\mathcal{H}_c$, i.e. $\mathcal{H}=\mathcal{H}_p\otimes\mathcal{H}_c$.
The positions Hilbert space is spanned by the points on a circle $\{\ket{x} | x \in \{0,\cdots,P-1\}\}$, while $\mathcal{H}_c$ is spanned by the two possible coin states $\{\ket{R},\ket{L}\}$, corresponding to heads and tails.
The evolution for one step of the QW is given by the unitary operator $$U=S\cdot (I_p\otimes R_c)$$
where $I_p$ is the identity operator in $\mathcal{H}_p$, $R_c \in SU(2)$ is a rotation in $\mathcal{H}_c$, which in matrix form we can write it as:
\begin{equation}\label{eq:coin-op}
R_c(\theta) = 	\left( 
				\begin{array}{cc}
					\cos(\theta) & \sin(\theta) \\
					-\sin(\theta) & \cos(\theta)
				\end{array}
			\right),
\end{equation} 
and
\begin{eqnarray}\label{eq:shift-op}
S=\sum_{x=0}^{P-1}\big[\ket{x+1 \!\!\pmod P}\bra{x}\otimes\ket{R}\bra{R}\\ 
+\ket{x-1 \!\!\pmod P}\bra{x}\otimes\ket{L}\bra{L}\big]\end{eqnarray}
is the shift operator that moves the walker one position to the right or to the left on the circle, depending on its coin state. Notice that, since we are on a circle, the $P$-th position is identified with the $0$ position.

\section{Quantum Walk Key-Distribution Scheme}
\label{sec:qkdscheme}
In this section we introduce a QKD scheme based on QWs. In this context, the key for the QKD is encrypted as the message in the public-key cryptosystem introduced in~\cite{vla:rod:mat:pau:sou:15}. This modification is motivated by the fact that QKD schemes are more flexible than public-key protocols, as the key can be used by both Alice and Bob to send or authenticate messages. Also, more post-processing techniques (e. g. privacy amplification) can be applied, since we have as input a random string and not a plaintext message. In the latter case we should be careful during the post-processing not to degrade the message (we are left with less techniques). Furthermore, in the case of information leakage, we can safely abort the protocol, while during message transmission it would be late for that. 

Our QW-QKD scheme is depicted in Fig.~\ref{fig:protQKD} and presented below. We assume that the key can be chosen among $P$ possible keys. We also assume that the QW can be chosen from a prefixed discrete set known by both parties.

%\begin{widetext}
\begin{center}
\begin{figure}[H]
  \centering
  \includegraphics[width=\textwidth]{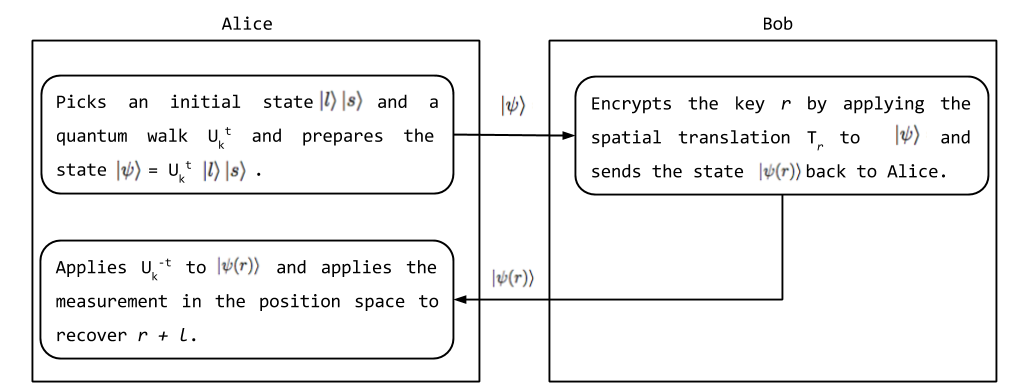}
\caption{Description of the basic steps of Protocol~\ref{prot:qkd}.                                     }
\label{fig:protQKD}
\end{figure}
\end{center}
%\end{widetext}

\begin{protocol} Quantum key-distribution scheme\
\label{prot:qkd} 
\begin{description}
\item[\hspace{6mm}{\bf Inputs for the protocol}]\
	\begin{itemize}
		\item Key: 

			$r\in \{0, \dots, P-1 \}$, i.e., a key of at most $\log P$ bits, chosen by Bob uniformly at random;

		\item Quantum state generation:\
		
		The QW operator $U_k$ with $k\in \mathcal K=\{1,2,\ldots,K\}$,
			the number of steps $t \in \mathcal{T} = \{ T_0, \dots, T_{max} \} \subset \mathbb{N},$
		and the initial state $\ket{l}\otimes\ket{s}$, where	$l \in \{ 0, \dots, P-1 \}, $
			$s \in \{R,L\}.$
	\end{itemize}
In the above, $U_k$, the QW operator is defined as $U_k=S\cdot (I_p\otimes R_c (\theta_k))$, where $S$ is the shift operator and $R_c(\theta_k)$ is a rotation of $\theta_k=k\cdot 2\pi/K$ in the coin space, see Eq.~\ref{eq:coin-op} and~\ref{eq:shift-op}.  \vspace{3mm}

\item[\hspace{6mm}{ \bf Quantum state generation} ]\
	\begin{itemize}
	\item Alice chooses uniformly at random $l \in \{0, \dots ,P-1 \} $ and $s \in\{R,L\}$, and generates the initial state $\Ket{l}\Ket{s}$.
	
	\item Then she chooses, also at random, the QW 
	$U_k=S\cdot \big(I_p\otimes R_c (\theta_k)\big)$ 
	and the number of steps $t\in\mathcal T$.   
	
	\item Finally, she generates the quantum state:
	$$
	\Ket{\psi} 	= U_k^t\Ket{l}\Ket{s} 
				= \big[S\cdot \big(I_p\otimes R_c (\theta_k)\big)\big]^t \Ket{l}\Ket{s},
	$$
	and sends it to Bob.
	\end{itemize}
\vspace{3mm}
\item[\hspace{6mm}{\bf Key encryption} ]\
	\begin{itemize}
	
	\item Upon obtaining the quantum state $\Ket{\psi}$	from Alice, Bob encrypts the key $r$ by applying spatial translation $T_r=\sum_{i=0}^{P-1}\ket{i + r \pmod{P}} \bra{i}$ to obtain: 
		\[\Ket{\psi(r)} = (T_r\otimes I_c) \Ket{\psi},\] 
		where $I_c$ is the identity operator in the coin space. 
	\item Bob sends $\Ket{\psi(r)}$ to Alice.
	\end{itemize}
\vspace{3mm}	
	\item[\hspace{6mm}{\bf Key decryption}]\
	\begin{itemize}
	\item Alice applies $U_k^{-t}$ to the state $\Ket{\psi(r)}$.

	\item She performs the position measurement 
	$$
	M = \sum_{i=0}^{P-1}\Ket{i} \Bra{i}\otimes I_c
	$$ 
	and obtains the result $i_0$. 
	\\The key sent by Bob is $r= i_0 - l \!\!\pmod {P}$.
\end{itemize}
\end{description}
\end{protocol}

It is clear, from the design of the protocol and the proof of correctness of the original quantum encryption scheme presented in \cite{vla:rod:mat:pau:sou:15} that, if no one interferes with the quantum states, then the protocol is correct and at the end, Alice and Bob will share a common string of length $\log P$, that they can use as a key. In the next section we prove the security of the protocol.

\subsection{Security of the Scheme}
In~\cite{vla:rod:mat:pau:sou:15}, the authors use the Holevo theorem to show that Eve can extract information about the key, by means of the quantum states $\ket{\psi}$ and $\ket{\psi(r)}$ that Alice and Bob exchange, only with negligible probability. Here, we do not present this proof of security for the sake of briefness, but the interested reader can find it in~\cite{vla:rod:mat:pau:sou:15} with all the details. There, it was shown that the protocol is secure if the size of the space of parameters, which are chosen uniformly at random, is exponentially large with respect to the length of the message transmitted (the key in our case of a QKD protocol), see Equation (19) and the comment below it, in Section 3.2 of~\cite{vla:rod:mat:pau:sou:15}. Moreover, for the protocol to be efficient, it was shown that the size of the $\mathcal T$ (set of possible number of steps) should be polynomial with respect to the length of the message/key transmitted, see Section 3.3 of~\cite{vla:rod:mat:pau:sou:15}. Therefore, to maintain the protocol's security, it is necessary that the size of $\mathcal K$, the set of possible coin unitaries $U_k$, is exponentially large with respect to the size of the key. As a consequence, while it is possible to fix the number of steps of the QW, and thus somewhat simplify its implementation, to maintain the security, it is necessary to keep the parameter $k$ of the protocol. 

Another type of attack that Eve can perform is a full man-in-the-middle attack, in which she impersonates Alice to Bob and vice versa, while they think that they are communicating directly. This attack gives Eve the chance to intercept and alter the communication between them. In public-key cryptosystems such attacks can be prevented by using a public-key infrastructure, which is assumed to work as a trusted third party. In our QKD modification though, such an assumption could not be used, therefore we should complete the security analysis of the scheme, by taking into account full man-in-the-middle attacks. To this end, we propose two different verification procedures, that allow Alice and Bob to verify that what they receive is actually coming from each other and not from an eavesdropper pretending to be either of them. We should note that for both verification methods, Alice and Bob need to share a classical public authenticated channel (a common requirement in QKD protocols, such as the well-known BB84 scheme~\cite{ben:bra:84}).

\subsubsection{Standard Verification}
\label{sec:stdver}
The first technique we propose is a standard cut-and-choose verification, which is achieved by adding redundancy to our scheme. Clearly, the verification is needed twice in our protocol: once when Alice sends the QW state to Bob and once when Bob sends the encoded key to Alice.\\ 

\textbf{Verification 1}: \textit{Bob verifies that it was Alice who sent him the quantum state.}

It is needed to prevent Eve from sending her choice of quantum states to Bob, which would allow her to read the encrypted key while he is sending it back to Alice. 
	\begin{itemize}
		\item Alice sends to Bob $\bigotimes_{i=1}^{m}\ket{\psi_i}$, that is, several quantum states $\ket{\psi_i}$, generated by a QW as described in the previous section. Each $\ket{\psi_i}$ is generated using independently chosen walk parameters and initial states $(k_i,t_i,l_i,s_i)$.
			
		\item After Bob receiving $\bigotimes_{i=1}^{m}\ket{\psi_i}$, Alice, through a classical authenticated channel, sends him a string $v=v_1v_2\ldots v_m$ of $m$ bits, such that $v_i=1$ if the corresponding $\ket{\psi_i}$ is going to be used for verification and $v_i=0$ otherwise, that is, if the corresponding $\ket{\psi_i}$ will be used by Bob to encode part of the key. Through the classical channel, she also sends $(j,k_j,t_j,l_j,s_j)$, for some uniformly at random chosen $j$'s that belong in the set $\{1,\ldots,m\}$. Let the number of these $j$'s be $m/3$.
		\item Bob verifies that for all these $j$'s, the received states $\rho_j=\ket{\psi_j}\bra{\psi_j}$ are indeed equal to the pure states 
		$$\Ket{\psi_j} 	= U_{k_j}^{t_{j}}\Ket{l_j}\Ket{s_j}. $$
		In order to verify that, he applies $U_{k_j}^{-t_j}$ to the states $\ket{\psi_j}$, for all $j$ and then performs a measurement for each $j$ in the positions space as well as in the spin space. This measurement (for each $j$) is described by the operator:
		\[
		\begin{array}{rcl}
		M_{l_j,s_j}&=&\ds \sum_{l_j,s_j}\alpha_{l_j,s_j}\ket{l_j,s_j}\bra{l_j,s_j}\\
		&=&\ds \sum_{l_j}l_j\ket{l_j}\bra{l_j}\otimes\sum_{s_j}s_j\ket{s_j}\bra{s_j}.
		\end{array}
		\]
		This way, he traces out all these $\ket{\psi_j}$'s and he is left with $2m/3$ quantum states. We call the reader's attention to the fact that if the verification fails for any $j$, the protocol is stopped.
 		\end{itemize}

\textbf{Verification 2}: \textit{Alice verifies that it was Bob who sent her the encrypted key.}\\

 This procedure is needed to prevent Eve from sending to Alice a message that would decrypt a key different from the one sent by Bob. In this case, Alice and Bob would not be able to communicate, while Eve would be able to decrypt messages sent by Alice (not vice versa). To prevent this from happening, Alice and Bob repeat the verification procedure 1, with the roles switched. In particular, the two are performing the following steps:
\begin{itemize}
	\item Bob encrypts $r_{i}\in \{0,\ldots,P-1\}$ in each of the states of the remaining product state $\bigotimes_{i=1}^{2m/3}\ket{\psi_i}$ as follows:
	 $$\Ket{\psi(r_{i})} = (T_{r_i}\otimes I_c) \Ket{\psi_{i}}, \forall i\in\{1,\ldots,2m/3\}$$
	 that is, translating each state $\ket{\psi_i}$ by $r_i$ in the positions space, leaving the spin part of the state unaltered. 
	 \item He sends the product state $\bigotimes_{i=1}^{2m/3}\Ket{\psi(r_{i})}$ to Alice.
	 \item Then he chooses $m/3 $ uniformly at random $j'$'s out of the $2m/3$ unused indices from the previous verification procedure. Through the classical public authenticated channel he sends a classical string $v'=v'_1v'_2\ldots v'_{2m/3}$ of $2m/3$ bits, such that $v'_i=1$ if the corresponding $\ket{\psi(r_i)}$ is going to be used for verification and $v'_i=0$ otherwise, that is, if the corresponding $\ket{\psi(r_i)}$ contains part of the key. For each $j'$ chosen (for which $v'_i = 1$), he also sends through the classical public authenticated channel the index and the respective $r_{j'}$'s used to generate the state $\ket{ \psi(r_i)}$. 	
	 \item In the last step, Alice applies $U_{k_i}^{-t_i}$ on the $2m/3$ states $\Ket{\psi(r_i)}$ and then, for each $i$, she performs a measurement on the positions space. Let the outcomes be denoted by $\alpha_i, i\in \{1,\ldots,2m/3\}$. For all the indices she computes $r_i=\alpha_i-l_i,$ where $l_i$ are the initial positions on the circle that she used for the generation of the quantum states $\ket{\psi_i}.$ Finally, for each $j', r_{j'}$ sent by Bob, Alice verifies the consistency of their results.
	 \item The key is given by the concatenation of the bits $r_i$ that were not used during the two verification procedures and it has $m\cdot (\log P)/3$ bits. Usually, the choice of $m$ is dependent on the desired length, $\log P$, of the key, and in order to make the success probability of a man-in-the-middle attack negligible on $\log P$, it is common to use $m=(\log P)/3$.
\end{itemize}

\subsubsection{Verification using maximally entangled states}
\label{sec:bellver}
In this section, we present an alternative verification procedure, which prevents Eve from trying to infer the key by first entangling her ancillas with the systems sent by Alice, and then performing an additional operation (say, a measurement) on the joint system of her ancillas and those carrying the encrypted key sent back to Alice by Bob; a method which in general would give her access to some non-negligible amount of information, so that Alice and Bob are not able to securely communicate. 
Note that this verification procedure could also be used against the previous attack in which Eve simply impersonates Alice to Bob, and vice versa. 

During the first step of the protocol (``Quantum state generation''), in addition to generating QW states
\begin{equation}
\label{states:1}
	\ket{\psi}_{qw}=U_k^t\ket{l}\ket{s}
\end{equation}
used to encode the key, for the verification purposes Alice also creates a number of Bell-like maximally entangled states
\begin{equation}
\label{states:2}
	\ket{\psi}_{qw}=\frac{1}{\sqrt{(\log 2P)!}}\sum_{i=0}^{2P-1}\ket{i}_a\ket{i}_{qw}.
\end{equation}
between the ancilla systems (denoted by $a$) and the QW systems (denoted by $qw$), each of dimension $2P$ (the dimension of the actual QW). At the end of the first step, Alice sends to Bob a random sequence of QW states, each either in the form $\ket{\psi}_{qw}$, or $\rho_{qw} = \Tr_a\ket{\psi}\bra{\psi}_{qw}$, while keeping the ancillas with her. Alice also sends through a classical public authenticated channel a classical string $v=v_1\ldots v_n$, where $v_i=0$ if the $i$-th system is going to be used for the encoding of the key, while $v_i=1$ if the $i$-th system is going to be used for verification.

The proportion of states used to obtain the key and used for the verification can be chosen in a similar way as in the previous case. Usually, the dimension of the total Hilbert space $2P$ is of the form $2^n$ which, in turn, is isomorphic to the Hilbert space resulting from the  tensor product of $n$ 2-dimensional Hilbert spaces, and thus this state can be written as the tensor product of $n$ standard two-qubit $\ket{\phi^+}$ Bell states.
 
After Bob receives the systems, he and Alice perform Bell-like measurements on the states meant for the verification and they observe a maximal violation of the Bell's inequalities, since those states are maximally entangled. This way, these states are traced out and Bob is left with the states~\eqref{states:1} in which he will encode the key (as previously).

The same procedure is repeated again, when Bob sends the encoded key to Alice. He will send a sequence of states, some of the form $(\hat{T}_r\otimes \hat{I}_c)U_k^t\ket{l}\ket{s}$, in which part of the key is encoded and some of the form (4) (with his ancillary system $\ket{i}_B$ maximally entangled to the system sent to Alice), which are going to be used for the verification, as explained above. In the end of the key decryption phase and if all the verifications were okay, Alice will concatenate the parts of the key to obtain the full key.

\subsection{Efficiency and quantum memory requirements}
In~\cite{vla:rod:mat:pau:sou:15} it has been proven that the protocol is efficient, i.e., it requires only polynomial time (on the length of the message, say $n$) to transfer $n$ bits of information encoded in $n+1$ qubits. By introducing the verification steps in this QKD scheme we increase the complexity of the system to $n^2$, in order to make the probability of eavesdropping negligible. However, we should notice that, out of this scheme, the size of the key that Alice and Bob share at the end is also increased to $n^2/3$, considering $m=n$. Therefore, the number of bits in the key is linear in the number of qubits sent to Bob. As a conclusion, our QKD scheme is efficient, since the complexity increased, but only polynomially.

As already mentioned in the Introduction, the lack of stable quantum memories is a major issue in quantum cryptography, since it is a practical constraint that is not likely to be solved, at least in the near future. Short-term quantum memories already exist, however it is not always straightforward to argue about the security of a protocol, relying on their existence. In our case, though, things are quite clear. If Eve does not interfere, Alice and Bob do not need quantum memories to execute the protocol, thus the key distribution is independent of such practical constraints. However, the presence of Eve and the need of verification for Alice and Bob introduce memory requirements for {\em all} the parties. 

Below, we present the memory requirements for the case of Section~\ref{sec:stdver}, noting that the case of Section~\ref{sec:bellver} is analogous. To conduct her attack, Eve needs a stable quantum memory, in order to keep the states she intercepted by Alice, while waiting for Bob to encrypt and send the key. Subsequently, she will encode it in Alice's states and send it to her. Also, in this scenario, Alice and Bob need a quantum memory, in order to perform the verification. They need to save the quantum states for some time, while waiting for the other party to send the classical information. Observe that Eve's memory should be more stable than Alice's and Bob's, as the time Eve needs to save the quantum states for, is clearly longer than the time that Alice and Bob need for the same purpose.

Hence, we conclude that our QKD scheme is  secure, as long as Alice and Bob have at least as powerful equipment as the adversary Eve. Obviously, if the adversary is technologically more advanced, then virtually any real-life implementation of a security protocol becomes potentially vulnerable.

%%%%%%%%%%%%%%%%%%%%%%%%%%%%%%%%%%%%%%%%%%%%%%%%%%%%%%%%%%%%%%%%%%%%%%%%%%%%%%%%%%%%%%%%%%%%%%%%%%%%%%%%%%%%%%%%
%%%%%%%%%%%%%%%%%%%%%%%%%%%%%%%%%%%%%%%%%%     One-Way 1st    %%%%%%%%%%%%%%%%%%%%%%%%%%%%%%%%%%%%%%%%%%%%%%%%%%
%%%%%%%%%%%%%%%%%%%%%%%%%%%%%%%%%%%%%%%%%%%%%%%%%%%%%%%%%%%%%%%%%%%%%%%%%%%%%%%%%%%%%%%%%%%%%%%%%%%%%%%%%%%%%%%%

\section{One-way quantum walk key-distribution protocol}
\label{sec:oneway}
In this section we present a one-way QKD protocol based on QWs and we prove that it is secure. First, we state it in its prepare-and-measure form. 
The protocol procedure is depicted in Figure~\ref{fig:prot1}.\\

%\begin{widetext}
\begin{center}
\begin{figure}[H]
  \centering
  \includegraphics[width=\textwidth]{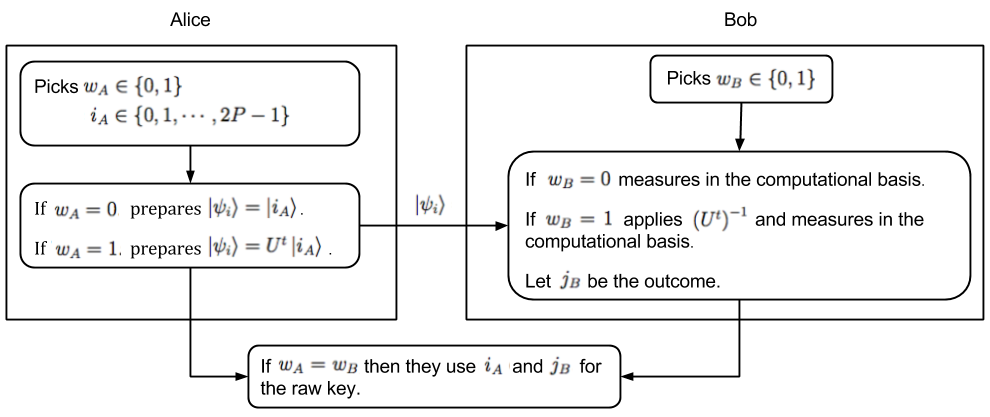}
\caption{Description of the basic steps of Protocol~\ref{alg:prot1}.}
\label{fig:prot1}
\end{figure}
\end{center}
%\end{widetext}

\begin{protocol}
\label{alg:prot1}
Let $\theta_k, t,$ and $P$ be publicly known where $P$ is the dimension of the position space of the QW, $t$ is the number of steps to perform the QW, and $\theta_k$ the coin parameter (see Equation~\eqref{eq:coin-op}).  Let $U_k$ be the QW operator $U_k = S\cdot (I_p\otimes R_c(\theta_k))$ that is also known by the parties (i.e., it is also publicly known) and let $F$ be an operator acting only on $\mathcal{H}_c$. $F$'s action is to ``flip'' the coin to some initial state before evolving the walk and is optional (in which case $F=I_c$). Finally, let $\ket{\psi_i} = U_k^t(I_p\otimes F)\ket{i}$ for $\ket{i} \in \mathcal{H}_p\otimes\mathcal{H}_c$.  We call the orthonormal basis $\{\ket{\psi_i}\}$ the QW basis and denote the computational basis by $Z$. This means that the QW basis is obtained from the computational basis $Z$, when performing the QW with respect to the initial state $\ket{i}$ or, in other words, the unitary operator that describes the basis change is the unitary of the QW.

The protocol consists of $N$ iterations of the following steps:
\begin{enumerate}
  \item Alice picks a random bit $w_A \in \{0,1\}$ and a value $i_A \in \{0,1, \cdots, 2P-1\}$.
    \begin{itemize}
    \item If $w_A = 0$: Alice will prepare and send to Bob the $2P$-dimensional state $\ket{\psi_i} = \ket{i_A}$.
    \item If $w_A = 1$: Alice will prepare and send to bob the $2P$-dimensional state $\ket{\psi_i} = U_k^t(I_p\otimes F)\ket{i_A}$.
    \end{itemize}

  \item Bob picks a random bit $w_B \in \{0,1\}$.
    \begin{itemize}
    \item If $w_B = 0$: Bob measures the received $2P$ dimensional state in the computational $Z$ basis resulting in outcome $j_B$.
    \item If $w_B = 1$: Bob measures in the QW basis (alternatively, he inverts the QW by applying $\left(U_k^t\right)^{-1}$ and measures the resulting state in the $Z$ basis).  The result is translated, in the obvious way, into an integer $j_B$.
    \end{itemize}
Note that he measures both the position and coin, as opposite to the previous protocol, where the measurement for the key was only on the positions space.
  \item Alice and Bob reveal, via the authenticated classical channel, their choice of $w_A$ and $w_B$.  If $w_A = w_B$, they will use their values $i_A$ and $j_B$ to contribute towards their raw key.  Otherwise, if $w_A \ne w_B$, they will discard this iteration.
\end{enumerate}

After the above process, Alice and Bob will use a cut-and-choose technique similar to Yao's ~\cite{yao:86}, to check eavesdropping by choosing a suitable subset of non-discarded iterations for parameter estimation in the usual manner (discarding those chosen iterations from the raw key). This allows them to estimate the disturbance $Q_Z$ and $Q_W$ in the $Z$ and QW bases respectively (i.e., in the absence of noise $Q_Z = Q_W = 0$).  If this disturbance is ``sufficiently low'' (to be discussed below) the users proceed with error correction and privacy amplification in the usual manner.

\end{protocol}
\subsection{Security}

In order to prove the security of Protocol~\ref{alg:prot1}, we will construct, in the usual way, an equivalent entanglement-based protocol~\cite{ben:bra:mer:92,lo:cha:99}. Proving security of this entanglement-based protocol will show the security of the prepare-and-measure version. This equivalence between entanglement-based and prepare-and-measure QKD protocols was first established by Bennett, Brassard and Mermin in~\cite{ben:bra:mer:92}. Since then, the relationship between the presence of entanglement (and specifically the ability of the involved parties to certify or distil entanglement) and the security of prepare-and-measure QKD protocols has been developed and thoroughly investigated~\cite{cur:lew:lut:04,cur:guh:lew:lut:05}. In this context, a quite common technique, when it comes to proving security of prepare-and-measure QKD protocols (such as the ones we present in this work), is to consider an equivalent entanglement-based protocol and prove its security~\cite{ben:bra:mer:92,lo:cha:99}. We should stress that it can also hold even in the case that the devices are not trusted (device-independent QKD), given that some specific assumptions about the devices are made~\cite{woo:pir:15}.\\

 For this entanglement-based version, for each one of the $N$ iterations, we make changes to steps (1) and (2), replacing them as follows:\vspace{5mm}

\textbf{New Step (1)}: Alice prepares the entangled state:
\[
\ket{\phi_0} = \frac{1}{\sqrt{2P}}\sum_{i=0}^{2P-1}\ket{i,i}_{AB}
\]
which lives in the $4P^2$ dimensional Hilbert space: $\left(\mathcal{H}_p\otimes\mathcal{H}_c\right)^{\otimes 2}$.  She sends the second half (the Bob portion of $\ket{\phi_0}$) to Bob while keeping the first half (the Alice portion) in her private lab. \vspace{5mm}

\textbf{New Step (2)}: Alice and Bob choose independently two random bits $w_A$ and $w_B$.  If $w_A = 0$, Alice will measure her half of the entangled state in the computational $Z$ basis; otherwise she will measure her half in the QW basis.  Similarly for Bob and $w_B$.  Let their measurement results in values be $i_A$ on Alice's side and $j_B$ on bob's side.

We now show the security of this entanglement-based version of the protocol.  In the following proof, we will initially make three assumptions:
\begin{enumerate}
  \item[\bf A1:] Alice and Bob only use those iterations where $w_A = w_B = 0$ for their raw key.

  \item[\bf A2:] Eve is restricted to collective attacks (those whereby she attacks each iteration of the protocol independently and identically, but is free to perform a joint measurement of her ancilla at any future time of her choosing).

  \item[\bf A3:] Eve is the party that actually prepares the states which Alice and Bob hold.
\end{enumerate}

Assumption A1 is made only to simplify the computation and may be discarded later (alternatively, one may bias the basis choice so that $w_A$ and $w_B$ are chosen to be $0$ with high probability, thus increasing the efficiency of the protocol as is done for instance for BB84 in~\cite{lo:cha:ard:05}).  Assumption A2 may be removed later using a de Finetti-type argument~\cite{ren:gis:kra:05,chr:kon:ren:09,ren:07} (in this paper, we are only concerned with the asymptotic scenario, so the key-rate expression we derive will not be degraded). Note that removing A2 gives us the security. Assumption A3 gives greater advantage to the adversary; if we prove security using A3, then the ``real-world'' case, where assumption A3 is not used, will certainly be just as secure, if not even more.

In light of A2 and A3, Alice, Bob, and Eve, after $N$ iterations of the protocol, hold a quantum state $\rho_{ABE}^{\otimes N}$, where $\rho_{ABE} \in \mathcal{H}_A\otimes\mathcal{H}_B\otimes\mathcal{H}_E$ with $\mathcal{H}_A \equiv \mathcal{H}_B \equiv \mathcal{H}_p\otimes\mathcal{H}_c$. We should note that, since we consider Eve to be an all-powerful adversary, there are no restrictions on the specific form of her Hilbert space $\mathcal{H}_E$. Following error correction and privacy amplification, Alice and Bob will hold a secret key of size $\ell(N)$.  Under the assumption of collective attacks (A2), we may use the Devetak-Winter key-rate expression~\cite{dev:win:05} to compute:
\[
r = \lim_{N\rightarrow \infty}\frac{\ell(N)}{N} = S(A|E) - H(A|B).
\]

Let $A_Z$ and $A_W$ be the random variables describing Alice's system, when she measures in the $Z$ or $QW$ basis, respectively.  Similarly, define $B_Z$ and $B_W$.  Under assumption A1, we are actually interested in the value:
\[
r = S(A_Z|E) - H(A_Z|B_Z).
\]

Computing $H(A_Z|B_Z)$ is trivial, given the observable probabilities:
\begin{equation}\label{eq:prot1:probZ}
p^Z_{i,j} = Pr(i_A=i \text{ and } j_B=j \text{ } | \text{ } w_A=w_B=0).
\end{equation}
The challenge is to determine a bound on the von Neumann entropy $S(A_Z|E)$.

To do so, we will use an uncertainty relation, proven in~\cite{ber:chr:col:ren:ren:10}, which states that for any density operator $\sigma_{ABE}$ acting on Hilbert space $\mathcal{H}_A\otimes\mathcal{H}_B\otimes\mathcal{H}_E$, if Alice and Bob make measurements using POVMs $\mathcal{M}_0 = \left\{M_x^{(0)}\right\}_x$ or $\mathcal{M}_1 = \left\{M_x^{(1)}\right\}_x$, then
\begin{equation}
S(A_0|E) + H(A_1|B) \ge \log\frac{1}{c},
\end{equation}
where
\begin{equation}
c = \max_{x,y}\left|\left| M_x^{(0)}M_y^{(1)} \right|\right|_\infty^2
\end{equation}
where we take $||\cdot||_\infty$ to be the operator norm and $A_i$ to be the random variable describing Alice's system after measuring $\mathcal{M}_i$ (we will later, similarly, define $B_i$). Assuming measurements $\mathcal{M}_0$ are used for key distillation, simple algebra, as discussed in~\cite{ber:chr:col:ren:ren:10}, yields the Devetak-Winter key-rate:
\begin{eqnarray*}
r = S(A_0|E) - H(A_0|B_0)& \ge &\log\frac{1}{c} - H(A_0|B_0) - H(A_1|B)\\[2mm]
&\ge & \log\frac{1}{c} - H(A_0|B_0) - H(A_1|B_1).
\end{eqnarray*}
The last inequality follows from the basic fact that measurements can only increase entropy.

In our case, we have $M_x^{(0)} = \ket{x}\bra{x}$ and $M_x^{(1)} = \ket{\psi_x}\bra{\psi_x}$ for $x \in \left\{0, 1, \cdots, 2P-1\right\}$.  Let $\ket{\psi_x} = \sum_{i=0}^{2P-1}\alpha_{x,i}\ket{i}$; then it is easy to see that for all $x,y$
\[
\left|\left|M_x^{(0)}M_y^{(1)}\right|\right|_\infty^2 = |\alpha_{y,x}|^2,
\]
and therefore
\begin{equation}\label{eq:prot1:cval}
c = \max_{x,y}|\alpha_{x,y}|^2,
\end{equation}
a quantity which depends exclusively on the choice of the QW parameters and not on the noise in the channel. Therefore, Alice and Bob should choose optimal $t, \theta_k$ and $P$ in order to minimize $c$ (thereby maximizing the key-rate equation).

As we show in the next section, this analysis is sufficient to derive good key-rate bounds.

\subsection{Evaluation}
As mentioned above, the value of $c$ depends solely on the QW parameters which are under Alice and Bob's control; therefore it is to their advantage to choose a QW which minimizes this value (i.e., such that, after evolving for $t$ steps, the probability of finding the walker at any particular position is small). 
 It is easy to see that, as $t \rightarrow \infty$, the values $|\alpha_{x,y}|$ do not converge to a steady state which is why, usually, one considers the time-averaged distribution when analyzing QWs on the cycle~\cite{aha:amb:kem:vaz:01,kem:03}. 
However, in our QKD protocol, we do not care what happens at large $t$; instead, we wish to find an optimal $t$ and one that is preferably not ``too large'' (the larger it is, the longer, in general, it might take Alice to prepare the state and Bob to reverse it). Notice that, while in the previous two-way protocol the choice of large $t$ increased the security of the protocol, here the role of $t$ is different. Larger $t$ does not mean that it is harder for Eve to distinguish.  It's just a parameter that Alice and Bob can tune to get optimal noise resistance, and in this context, small $t$ (even $t=1\ \text{or}\ 2$) can lead to secure systems, as well.  Different values of $t$ correspond to different noise tolerances and  our investigation is looking for optimal values of $t$ where ``optimal'' means the highest theoretical noise tolerance for a \textit{given} walk setting. If our system were to be implemented in practice, ``optimal'' would probably take a much different form and Alice and Bob would have to consider their device imperfections, etc.  Such an investigation, though very interesting, is outside the scope of this paper, but it presents a relevant direction of future work.

We begin by looking at various walk parameters and finding the minimal value of $c$ when $F=I_c$, the identity operator. Note that, on the circle, it makes sense only to consider odd $P$ as even $P$ would force the support of the probability amplitudes onto even or odd numbered nodes only thereby increasing the overall value of $|\alpha_{x,y}|$. We wrote a computer program to simulate the walk for time steps $t = 1, 2,\cdots , T_{\text{max}}$ (for user-specified value $T_{\text{max}}$) searching for the optimal value of $t$ (i.e., a value for $t$ whereby $c$ is minimum). For the evaluation we used a more general form of the coin rotation operator:

\begin{equation*}
R_c(\theta,\phi) = 	\left( 
				\begin{array}{cc}
					e^{i\phi} \cos(\theta) & e^{i\phi} \sin(\theta) \\
					-e^{-i\phi} \sin(\theta)& e^{-i\phi} \cos(\theta)
				\end{array}
			\right),
\end{equation*}

The results for $\theta = \pi/4, \phi = 0$, and for various $P$ are shown in Figure~\ref{fig:prot1:cvals}.
	
\begin{figure}[H]
  \centering
  \includegraphics[scale = 0.6]{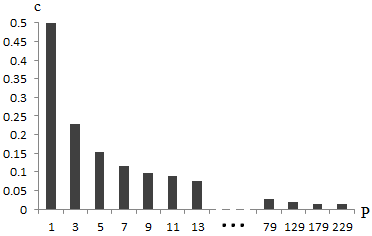}
\caption{Showing minimal value of $c$ found by our program for given position space dimension $P$ when $\theta = \pi / 4,\phi=0$ and $F=I_c$.  When $P \le 13$ we set $T_{\max} = 5000$; when $P \ge 79$ we set $T_{\max} = 50000$.  Note that, the smaller $c$ is, the better for Alice and Bob.  Note also that $P$ is the dimension of the position space, \emph{not} the number of qubits sent which would actually be $\lceil\log P\rceil+1$ (where the extra ``$+1$'' is due to the coin).}\label{fig:prot1:cvals}
\end{figure}

Now that we can find the optimal choice of QW parameters for particular values of $P$ and, more importantly for our work here, the resulting value of $c$. To this end, we have to compute our bound $r$ and determine for what noise levels we can have $r > 0$. In practice, one would observe values $p_{i,j}^Z$ and $p_{i,j}^W$ (see Equation~\eqref{eq:prot1:probZ} and define $p_{i,j}^W$ analogously) and use these to directly compute $H(A_Z|B_Z)$ and $H(A_W|B_W)$ as required by the key-rate equation.  For the purpose of illustration in this paper, however, we will evaluate our key-rate bound assuming a generalized Pauli channel as discussed in~\cite{bae:aci:07} (see, in particular, Section 7 of that source). This channel maps an input state $\rho$ to an output state $\mathcal{E}(\rho)$ defined as:
\begin{equation}\label{eq:prot1-channel}
\mathcal{E}(\rho) = \sum_{m=0}^{2P-1}\sum_{n=0}^{2P-1}p_{m,n}\;\mathcal{U}_{m,n}\;\rho\;\mathcal{U}_{m,n}^*,
\end{equation}
where
\begin{equation}
\mathcal{U}_{m,n} = \sum_{k=0}^{2P-1}e^{{\pi \cdot i \cdot  k \cdot n}/{P}}\ket{k+m}\bra{k}.
\end{equation}
That is, this channel $\mathcal{E}(\cdot)$ models an adversary's attack which induces phase and flip errors with probabilities denoted by $p_{m,n}$.  In our numerical computations to follow, we will use:
\begin{equation}\label{eq:prot1-channelpr}
p_{i,j} = \left\{
\begin{array}{ll}
	1-E_r & \text{ if } i = j = 0\\[4mm]
	\ds \frac{E_r}{(2P)^2-1} & \text{ otherwise}
\end{array}
\right. .
\end{equation}
It is clear that $\sum_{i,j}p_{i,j} = 1$.  Furthermore, when $E_r=0$, we have $\sum_ip_{i,i}^Z = \sum_ip_{i,i}^W = 1$ (i.e., there is no disturbance in the channel) while as $E_r$ increases, the disturbance also increases.

Finally, we define the total noise in the channel to be:
\[
Q = \sum_{a\ne b}p_{a,b}^Z = \sum_{a\ne b} Pr \big( A_Z=a \text{ and } B_Z = b \text{ } | \text{ } w_A = w_B = 0 \big).
\]
That is to say, $Q$ represents the quantum error rate (QER) of the channel.

The maximally tolerated QER, for those QWs analyzed in Figure \ref{fig:prot1:cvals}, and using the above described noise model, is shown in Figure \ref{fig:prot1:maxnoise1}.  Note that, when $P=1$ and $t=1$, we recover the BB84 limit of $11\%$ which is to be expected since, with these choice of parameters, we are essentially running the BB84 protocol.
\begin{figure}[H]
  \centering
  \includegraphics[scale=0.6]{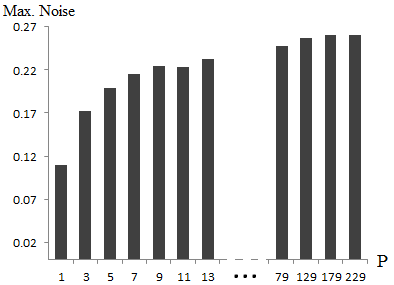}
\caption{Showing the maximally tolerated noise level for our protocol using parameters found in Figure \ref{fig:prot1:cvals} and using the quantum channel described by Equations~\eqref{eq:prot1-channel} and~ \eqref{eq:prot1-channelpr}. The lack of increase in noise tolerance from $P=9$ to $P=11$ (while other choices caused an increase) indicates that $T_{\max}$ was too low.  Note that, when $P = 1$, we recover the BB84 tolerance of $Q = 0.11$ as expected.  Also note that, when $P = 229$, the maximal tolerated noise is $Q = 0.261$.}\label{fig:prot1:maxnoise1}
\end{figure}
Observe in Figure \ref{fig:prot1:maxnoise1} that there is a lack of increase when $P=9$ and $P=11$; this indicates that our choice of $T_{\max} = 5000$ was too low.  Running our simulator again with $T_{\max} = 50000$ for these small $P$ values yields a maximally tolerated noise level shown in Figure \ref{fig:prot1:maxnoise2}.

\begin{figure}[H]
  \centering
  \includegraphics[scale=0.60]{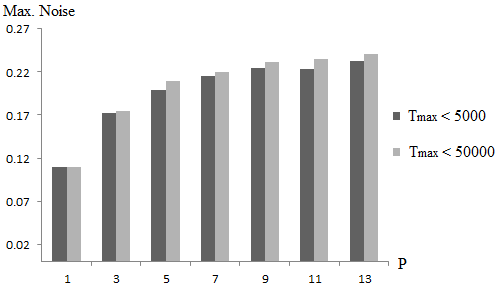}
\caption{Comparing the maximally tolerated noise when $t$ is allowed to be as large as $50000$ (light gray) or only $5000$ (dark gray); again when $F = I$ and $\phi = 0$.  In this case, when $P=13$ and $T_{\max} = 50000$, the maximal tolerated noise $Q$ is $Q = 0.241$.}\label{fig:prot1:maxnoise2}
\end{figure}

Finally, we re-run the simulator, using $T_{\max} = 5000$ and $T_{\max} = 50000$ for a different QW parameter of $\theta = \sqrt{2}\pi/4$ which, for these particular upper-bounds on $t$ yield a higher tolerated noise as shown in Figures \ref{fig:prot1:walk2} and \ref{fig:prot1:walk2-2}.  
We comment that, if $T_{\max}$ were larger, the two QWs may produce a QKD protocol with the same tolerated noise; however for these ``smaller'' bounds on $t$ the QW with parameter $\theta = \sqrt{2}\pi/4$ produces a more secure protocol than when $\theta = \pi/4$.  Since smaller $t$ implies a more efficient protocol, this is an advantage.  This opens two very interesting questions: first, do these QWs produce equivalent noise tolerances as $T_{\max}\rightarrow \infty$? Second, what other values of $\theta$ produce even more secure QKD protocols for small $T_{\max}$?  We comment that we also ran this numerical experiment for $\theta = \pi/5$ and $\theta = \pi/3$ but got worse noise tolerances.

\begin{figure}[H]
  \centering
  \includegraphics[scale = 0.60]{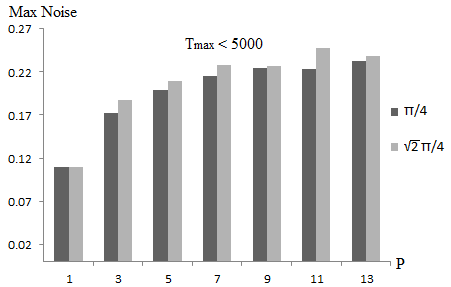}
\caption{Comparing the maximal tolerated noise levels of the QKD protocol when $\theta = \pi/4$ (dark gray) and $\theta = \sqrt{2}\pi/4$ (light gray).  In this chart, $T_{\max} = 5000$ which, observing the ``drop'' in tolerated noise when $P$ goes from 11 to 13, is too small.  See also Figure \ref{fig:prot1:walk2-2} for the same chart when $T_{\max} = 50000$.}\label{fig:prot1:walk2}
\end{figure}

\begin{figure}[H]
  \centering
  \includegraphics[scale = 0.60]{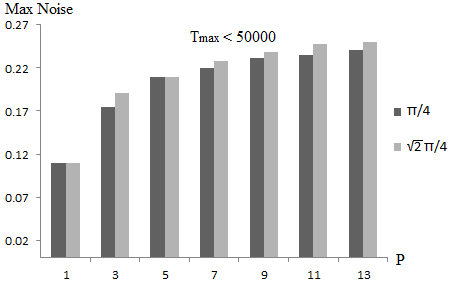}
\caption{Comparing the maximal tolerated noise levels of the QKD protocol when $\theta = \pi/4$ (dark gray) and $\theta = \sqrt{2}\pi/4$ (light gray).  In this chart, $T_{\max} = 50000$.  In all cases, the QW parameter $\theta = \sqrt{2}\pi/4$ produces a more secure QKD protocol for this upper-bound on $t$.  Note that, as $T_{\max} \rightarrow \infty$, they may produce equally secure protocols; this, as discussed in the text, is an open question.  In this case, when $P=13$ and $\theta = \sqrt{2}\pi/4$, the maximally tolerated noise is $0.25$ (compared to $0.241$ when $\theta = \pi/4$).}\label{fig:prot1:walk2-2}
\end{figure}

From the above it is clear that careful choice of the QW parameters is vital for producing a QKD protocol tolerant of high noise channels.  To investigate this further, we simulate the QW for all $\theta,\phi \in \{k\pi/10 \text{ } | \text{ } k = 0, 1, \cdots, 10\}$.  Furthermore, for each setting, we also consider the use of $F = I, F = X$, and $F = Y$, where:
\[
\begin{array}{clc}
X = \ds \frac{1}{\sqrt{2}}\left(
			\begin{array}{cc}
					1 & 1\\
					1 & -1
			\end{array}\right)
& &
Y = \ds \frac{1}{\sqrt{2}}\left(
			\begin{array}{cc}
					1 & 1\\
					i & -i
			\end{array}\right).
\end{array}
\]

For each setting, we find the optimal choice of time $t \le 5000$ which produces a minimal $c$.  We then take this value and determine the highest disturbance the resulting protocol can withstand.  The respective data is summarized in Table \ref{table:sim-results}.

\begin{center}
\begin{table}[H]
\centering
{\footnotesize
\begin{tabular}{c|ccccc|ccccc|ccccc|}
&\multicolumn{5}{|c|}{$F = I$} & \multicolumn{5}{|c|}{$F = X$} & \multicolumn{5}{|c|}{$F = Y$}\\
\cline{2-16}
P
& $\theta$ & $\phi$ & $t$ & $c$ & $Q_{\text{max}}$
& $\theta$ & $\phi$ & $t$ & $c$ & $Q_{\text{max}}$
& $\theta$ & $\phi$ & $t$ & $c$ & $Q_{\text{max}}$\\
\hline
3 & $0.4\pi$ & $0.2\pi$ & 4584 & $0.171$ &$0.220$
  & $0.8\pi$ & $0.8\pi$ & 3994 & $0.181$ & $0.211$
  & $0.7\pi$ & 0 & 1502 & $0.167$ &$0.225$\\
5 & $0.7\pi$ & $\pi$   & 4340 & $0.147$ &$0.205$
  & $0.9\pi$ & $0.5\pi$ & 3870 & $0.132$ & $0.22$
  & $0.3\pi$ & 0 & 3748 & $0.106$ & $0.253$\\
7 & $0.6\pi$ & $0.9\pi$ & 3946 & $0.088$ & $0.252$
  & $0.7\pi$ & $0.8\pi$ & 3391 & $0.099$ & $0.236$
  & $0.3\pi$ & $0.5\pi$ & 1275 & $0.083$ & $0.261$\\
9 & $0.6\pi$ & $0.6\pi$ & 1269 & $0.077$ & $0.252$
  & $0.9\pi$ & $0.7\pi$ & 3041 & $0.079$ & $0.250$
  & $0.3\pi$ & $0.5\pi$ & 965 & $0.069$ & $0.267$\\
11& $0.6\pi$ & $0.4\pi$ & 1221 & $0.069$ & $0.252$
  & $0.8\pi$ & $0.4\pi$ & 481 & $0.0724$ & $0.245$
  & $0.7\pi$ & $0.5\pi$ & 277 & $0.054$ & $0.284$
\end{tabular}
}
\caption{Showing the optimal choice of QW parameters to maximize the noise tolerance ($Q_{\text{max}}$) of the resulting protocol.  For this data, we searched for QWs with at most $T_{\text{max}} = 5000$ steps and with parameters $\theta,\phi \in \{k\pi/10 \text{ } | \text{ } k = 0, 1, \cdots, 10\}$.}\label{table:sim-results}
\end{table}
\end{center}

Note that, for some data points (e.g., when $P = 5$ and $F = I$) there is a drop in the maximum tolerated noise.  This is a consequence either of setting $T_{\text{max}}$ too small, or we need to simulate more QW parameters (as is done in Table \ref{table:sim-results2}). For example, when we set $T_{\max} = 50000$, for $P = 5$ and $F = I$, we get a maximum noise tolerance of $0.236$ when $t = 40847$.  Note also, that setting $F = Y$ achieves the best result for this test, $Q_{\text{max}}=0.284$.\\

In Table 2, we carried out the same experiment, however this time searching over QW parameters in the set $\theta, \phi \in \{k\pi/20 \text{ } | \text{ } k = 0, 1, \cdots, 20\}$. Again, the best result for this case is $Q_{\text{max}}=0.284$ and is achieved when considering $F = Y$.

\begin{center}
\begin{table}[h!]\includegraphics[scale=0.99]{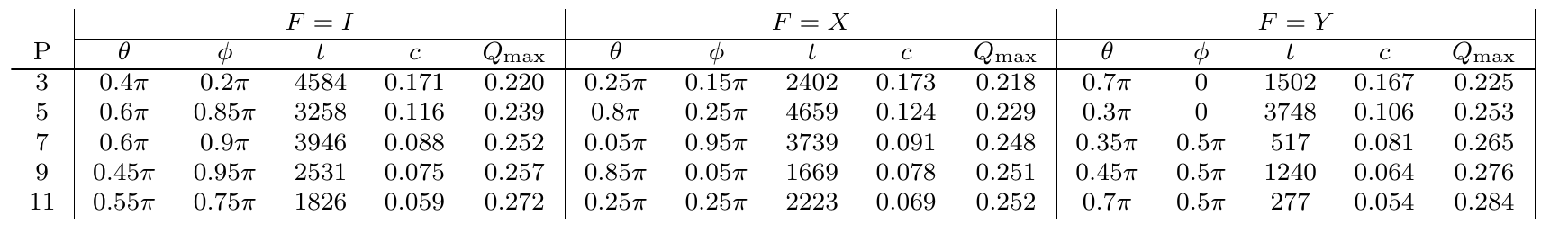}	
\caption{Showing the optimal choice of QW parameters to maximize the noise tolerance ($Q_{\text{max}}$) of the resulting protocol.  For this data, we searched for QWs with at most $T_{\text{max}} = 5000$ steps and with parameters $\theta,\phi \in \{k\pi/20 \text{ } | \text{ } k = 0, 1, \cdots, 20\}$.}
\label{table:sim-results2}
\end{table}\end{center}

As mentioned at the beginning of this Section, all the numerical results were obtained by simulating the evolution of the QW on a custom QW simulator that we wrote. However, we also verified the results through an alternative technique, namely by computing the probability amplitudes of the QW using the standard Fourier method (see, e.g.~\cite{nay:vis:00,ven:and:12}) of analyzing QWs. The results obtained by both methods agree with each other.

We would also like to stress that, while the use of high-dimensional walks is ``ideal'' from a noise-tolerance perspective, this is not required -- indeed, even with very small dimensions our protocol can tolerate the same level of noise as  BB84. In particular, notice in Figure~4, that already for the minimum $P=1$ we obtain the BB84 noise tolerance, which increases further for relatively  small dimensions $P=3, P=5$ and so on. We also simulated higher dimensions to investigate the theoretical properties of our systems for different walks.  However, such high dimensions are not required to get a robust protocol. Therefore, while the practical implementation of a high-dimensional protocol might be quite complex, our protocol is robust even for relatively small dimensions, thus the experimental challenges in a potential practical implementation could be reduced. Moreover, we should emphasise that recently there has been a remarkable progress in the experimental generation and manipulation of high-dimensional entangled states. In particular,  in~\cite{wang:etal:17,jha:etal:08,martin:etal:17,dada:etal:11,krenn:etal:14} the experimental generation of bipartite high-dimensional entangled states has been demonstrated, while in~\cite{bavaresco:etal:17,lavery:etal:12,mirhosseini:etal:13,islam:etal:17} there have been proposed experimental techniques for performing measurements in such states. Furthermore, the generation of multi-partite high-dimensional entangled states has been reported~\cite{erhard:etal:17,malik:etal:16,hiesmayr:etal:16}, as well as the generation of big arrays of such states~\cite{krenn:etal:16}.

Finally, we note that the protocol's security is not compromized by considering the existence or not of quantum memories. It is sufficient to consider the prepare-and-measure form of the protocol. Eve needs a quantum memory to perform her attack, as she needs to save her ancillary system throughout the execution of the protocol. In the contrary, the secure key distribution between Alice and Bob does not require any quantum memory. Therefore, if Eve does not have a quantum memory she cannot attack, while if even she has one and attacks, Alice and bob can defend against it and securely share a key at the end. Notice, that even if we consider the entanglement-based version of the protocol, again the security is independent of any quantum memory requirements, as Eve for her attack needs a more stable quantum memory than Alice and Bob need to defend against it and securely distil the key.

%%%%%%%%%%%%%%%%%%%%%%%%%%%%%%%%%%%%%%%%%%%%%%%%%%%%%%%%%%%%%%%%%%%%%%%%%%%%%%%%%%%%%%%%%%%%%%%%%%%%%%%%%%%%%%%%
%%%%%%%%%%%%%%%%%%%%%%%%%%%%%%%%%%%%%%%%%%     Semi-Quantum      %%%%%%%%%%%%%%%%%%%%%%%%%%%%%%%%%%%%%%%%%%%%%%%
%%%%%%%%%%%%%%%%%%%%%%%%%%%%%%%%%%%%%%%%%%%%%%%%%%%%%%%%%%%%%%%%%%%%%%%%%%%%%%%%%%%%%%%%%%%%%%%%%%%%%%%%%%%%%%%%

\section{Semi-Quantum Key-Distribution Scheme}
\label{sec:semi-qkdscheme}

In this section we will present a SQKD protocol based on QWs. Semi-quantum protocols can be seen as practical instances of QKD, since they involve less quantum hardware, as one of the parties is completely classical. Also, they are interesting from a theoretical point of view, as they can be treated as a measure of the ``quantumness'' needed for a protocol to surpass the security of its classical counterparts. It is assumed that one of the parties, e.g. Alice, is fully quantum, i.e., possesses quantum equipment, while the other party, e.g. Bob, is restricted to classical operations only, and for that reason is usually called the {\em classical party}. 

Semi-quantum protocols rely on a two-way quantum channel allowing a quantum state to travel from Alice to Bob, then back to Alice.  When first introduced by Boyer {\em et al.} in \cite{boy:ken:mor:07}, these classical operations involved Bob either measuring the incoming qubit in the $Z = \{\ket{0}, \ket{1}\}$ basis, or reflecting the incoming qubit, bouncing it back to Alice undisturbed.  For our purposes, we extend this definition of ``classical'' operations to operate with higher dimensional systems.  As we do not want to restrict ourselves necessarily to qubit encodings (and thus, dimensions that are powers of two), we will say that Bob, on receipt of an $D$ dimensional quantum state $\ket{\psi},$ may choose to do one of two operations:
\begin{enumerate}
  \item Measure and Resend: Bob may subject the $D$-dimensional quantum state to a measurement in the computational basis spanned by states: $\{\ket{0}, \ket{1}, \cdots,$ $ \ket{D-1}\}$.  He will then prepare a new $D$-dimensional quantum state in this same computational basis based on the result of his measurement.  Namely, if he observes $\ket{r}$ for $r \in \{0, 1, \cdots, D-1\}$, he will send to Alice the quantum state $\ket{r}$.
  \item Reflect: Bob may ignore the incoming $D$-dimensional quantum state and reflect it back to Alice.  In this case he learns nothing about its state.
\end{enumerate}

With these restrictions on the part of the classical user defined, we now depict and describe our protocol:

 % \begin{widetext}
    \begin{center}
    \begin{figure}[H]
    \centering
    \includegraphics[width=0.8\textwidth]{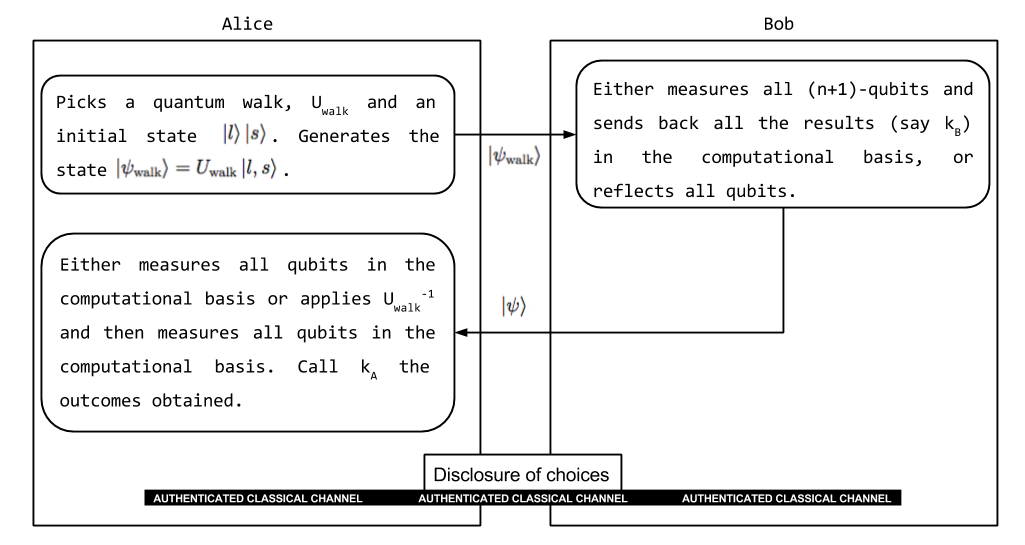}
    \caption{Description of the basic steps of Protocol~\ref{prot:semi-q}.}
    \label{fig:semi_QKD}
    \end{figure}
    \end{center}
 %   \end{widetext}

\begin{protocol} Semi-quantum key-distribution scheme\
\label{prot:semi-q} 
\begin{description}
\item[\hspace{6mm}{\bf Inputs for the protocol} ]\
\begin{itemize}
	\item $\ket{l,s}$, the initial state of the QW, where $l \in \mathcal L=\{ 0, \dots, P-1 \} $ is the initial position of the walker, and $s \in \mathcal S=\{R,L\}$ gives the initial coin state.
	\item $U_{\textnormal{walk}} = (U_k)^t \in \mathcal Q$, the evolution of the QW, where $k \in \mathcal K  =\{1,2,\ldots,K\}$ is the choice of a single step unitary $U_k$, and $t \in \mathcal{T} = \{ T_0, \dots, T_{max} \}$ is the number of steps of the QW. Thus, $\mathcal Q$ is the set of all possible QWs. Note that $\mathcal Q$  is publicly known.
\end{itemize}

\vspace{3mm}
\item[\hspace{6mm}{\bf Quantum state Generation}]\
\begin{itemize}
 \item Alice chooses uniformly at random $l \in \mathcal L=\{0,1,\cdots, P-1\}$ and $s\in\mathcal S=\{R,L\}$.  She also chooses a random QW operator $U_{\textnormal{walk}} \in \mathcal{Q}$ according to a publicly known distribution (e.g., uniform).  She then prepares the following state:

$$\ket{\psi_{\textnormal{walk}}} = U_{\textnormal{walk}}\ket{l,s}.$$
\item Alice sends this state to Bob.
\end{itemize}

\vspace{3mm}
 \item[\hspace{6mm}{ \bf Classical operations by Bob} ]\

Bob chooses either to measure-and-resend the quantum state in the computational basis $\{\ket{0}, \ket{1}, \cdots, \ket{2P-1}\}$ (note that, in this protocol as well, he measures both the position and coin in order to obtain the key, thus his measurement, and subsequent preparation, is of dimension $2P$); or he will reflect the quantum state back to Alice.

\vspace{3mm}

  \item[\hspace{6mm}{\bf Alice's final step} ]\
  
  Alice chooses one of the following two options:
  \begin{itemize}
  	\item She measures the returning quantum state in the computational basis and saves the result as $\kappa_A$. 
  	\item She first applies the inverse QW, $U_{\textnormal{walk}}^{-1}$, and then measures in the computational basis. Note that, in the absence of noise, if Bob reflects, her measurement outcome should be $\ket{l,s}$.
   \end{itemize}

\vspace{3mm}
  \item[\hspace{6mm}{\bf Disclosure}]\
   
  Alice discloses her choice of operation and Bob discloses his choice either to measure and resend or reflect.
  
  \vspace{3mm}
   \item[\hspace{6mm}{\bf Iterations} ]\
   
  The above process is repeated $N$ times. 
  
  \vspace{3mm}
   \item[\hspace{6mm}{\bf Results} ]\
  \begin{itemize}
  	\item Every time bob measures and resends and Alice measures in the computational basis, the parties add $1+\log P$ bits  to their final raw key. After $N$ iterations, the raw key will consist of $N(1 + \log P)/4$ bits, on average. Considering that Alice and Bob choose independently, uniformly and at random their actions at each iteration, then -- on average -- $N/4$ iterations will contribute to the raw key.
  	\item Every time Bob reflects and Alice measures after applying the inverse QW, the outcome of her measurement $(l_m,s_m)$ should be what she initially used to generate the QW state (i.e., it should be that $l=l_m$ and $s=s_m$). These iterations, together with some randomly chosen iterations of the first type (where Bob measures and resends), are used for error detection. 
  	\item The other iterations are discarded.
  \end{itemize}
\end{description}
\end{protocol}

\subsection{Proof of robustness}

As with the first protocol we proposed in this paper, the reliance on a two-way quantum channel greatly complicates the security analysis.  It was only recently that several SQKD protocols were proven secure \cite{krawec2015security,krawec2016security,krawec2016quantum,zha:qiu:mat:16}. However, the proof techniques developed in those works assumed qubit-level systems.  In our case, not only must we contend with a two-way channel, but also with the fact that the quantum states traveling between Alice and Bob are of dimensions higher than $2$.  This leads to significant challenges in any security analysis.

Therefore, to analyze the security of this protocol, we will prove that it is \emph{robust} as defined in ~\cite{boy:ken:mor:07,boy:gel:ken:mor:09}. That is, for any attack which Eve may perform which causes her to gain information on the raw key, this attack must necessarily lead to a disturbance in the channel which can be detected with non-zero probability by Alice and Bob.  Typically proving robustness is a first-step in the security analysis of semi-quantum cryptographic protocols.

\begin{theorem}\label{thm:sqkd-robust}
If $I \in \mathcal{Q}$ (where $I$ is the identity operator on the joint $2P$ dimensional system) and if, for every $(l,s), (l',s') \in \{0,1,\cdots, P-1\}\times\{R,L\}$ there exists a $U_{\textnormal{walk}} \in \mathcal{Q}$ and initial state $\ket{l_0,s_0}$ (all possibly depending on the choice of $(l,s)$ and $(l',s')$) such that $\braket{l,s|U_{\textnormal{walk}}|{l_0,s_0}} \ne 0$ and $\braket{l',s'|U_{\textnormal{walk}}|l_0,s_0} \ne 0$, then the SQKD protocol based on QWs is robust.
\end{theorem}
\begin{proof}

We will assume, similarly to~\cite{zou:qiu:li:wu:li:09,kra:14}, that Alice sends each (in our case $2P$-dimensional) quantum state, only after she receives one from Bob (excepting, of course, the first iteration).  In this case, Eve's most general attack consists of a collection of unitary operators $\left\{(U_F^{(i)}, U_R^{(i)})\right\}_{i=1}^N$ where, on iteration $i$ of the protocol, she applies $U_F^{(i)}$ in the forward channel (as the quantum state travels from Alice to Bob) and $U_R^{(i)}$ in the reverse channel.  These operators act on the $2P$-dimensional quantum state and Eve's private quantum memory.  We make no assumptions about how these operators are chosen -- for instance, Eve may choose them ``on the fly''; that is, she may choose operator $U_F^{(2)}$ after attacking with $U_F^{(1)}$.

Consider the first iteration $i=1$.  We assume, without loss of generality, that Eve's quantum memory is cleared to some pure ``zero'' state, denoted by $\ket{\chi}$, known to her.

In the remainder of this proof, we will treat the position space and the coin space as a single space $\Sigma$ of dimension $2P$.

We may describe the action of $U_F^{(1)}$ on basis states as follows
\[
U_F^{(1)}\ket{i, \chi} = \sum_{j=0}^{2P-1}\ket{j, e_i^j},
\]
where $\ket{e_i^j}$ are arbitrary states in Eve's ancillary system. These states are not necessarily normalized nor orthogonal; the unitarity of $U_F^{(1)}$ imposes some restrictions on them which we will use later.

With non-zero probability, this iteration may be used for error detection.  It is also possible that Alice chose to use $I \in \mathcal{Q}$ in this iteration and, thus, she sends the quantum state $\ket{\sigma}$ to Bob, for $\sigma \in \Sigma$. Furthermore, Bob chooses to measure and resend with non-zero probability. Therefore, to avoid detection, it must be that $\ket{e_i^j}\equiv 0$ for all $i \ne j$, and the unitarity of  $U_F^{(1)}$ yields $\braket{e_i^i|e_i^i} = 1$ for all $i$.  Thus:
\[
U_F^{(1)} \ket{i,\chi} = \ket{i,e_i^i}, \forall i=0,1,\cdots,2P-1.
\]
Now, consider $U_R^{(1)}$, the attack applied in the reverse channel.  We may write its action as follows:
\[
U_R^{(1)}\ket{i, e_i^i} = \sum_{w=0}^{2P-1}\ket{w, e_{i,i}^w}.
\]
The same argument as before applies: in particular, with non-zero probability Alice and Bob will use this iteration to check for errors, and so it must be that $\ket{e_{i,i}^w} \equiv 0 $ for $i \ne w$.  Thus
\[
U_R^{(1)}\ket{i, e_i^i} = \ket{i, e_{i,i}^i} = \ket{i, f_i}, \forall i = 0,1,\cdots, 2P-1,
\]
where we defined $\ket{f_i} \equiv \ket{e_{i,i}^i}$ for ease of notation.

Now, assume that Alice chooses a QW operator $\walkop \in \mathcal{Q}$, with $\walkop \ne I$.  Let $\ket{\sigma}$ be the initial state she prepares ($\sigma$ chosen at random from $\Sigma$).  In this case, the quantum state she sends to Bob may be written as:
\[
\walkop\ket{\sigma} = \ket{\psi_\sigma} = \sum_{i=0}^{2P-1}\alpha_i\ket{i}.
\]
Assume that $\walkop$ is chosen so that at least two of the $\alpha_i$'s are non-zero (such QWs exist by hypothesis).  If Bob reflects, the qubit state arriving at Alice's lab, after Eve's attack on both channels, is
\begin{equation}\label{eq:wk:walkresult}
U_R^{(1)}U_F^{(1)}(\walkop\otimes I_E)\ket{\sigma,\chi} = \sum_i \alpha_i\ket{i,f_i},
\end{equation}
where $I_E$ is the identity operator on Eve's ancilla.

Alice will subsequently apply the inverse QW operator and measure the resulting state, expecting to find $\ket{\sigma}$.  This is equivalent to her measuring in the QW basis $\{\ket{\psi_0}, \ket{\psi_1}, \cdots, \ket{\psi_{2P-1}}\}$, where $\ket{\psi_i} = \walkop\ket{i}$, and expecting to observe $\ket{\psi_\sigma}$.  In this QW basis, we clearly have
\[
\ket{i} = \sum_{j=0}^{2P-1}\braket{\psi_j|i}\ket{\psi_j},
\]
from which, we may write Equation~\eqref{eq:wk:walkresult} as:
\begin{eqnarray}
\sum_{i=0}^{2P-1}\alpha_i \left( \sum_{j=0}^{2P-1}\braket{\psi_j|i}\ket{\psi_j}\right) \otimes \ket{f_i}\\\\
=\sum_{j=0}^{2P-1}\ket{\psi_j}\otimes\left(\sum_{i=0}^{2P-1}\alpha_i\braket{\psi_j|i}\ket{f_i}\right).
\end{eqnarray}
Let $p$ be the probability that this iteration does not result in an error -- i.e., the probability that Alice measures $\ket{\psi_\sigma}$. From the above equation:
\[
p = \left| \sum_{i=0}^{2P-1}\alpha_i\braket{\psi_\sigma|i}\ket{f_i} \right|^2.
\]
Noticing that $\braket{\psi_\sigma|i} = \alpha_i^*$ (since $\ket{\psi_\sigma} = \sum_i\alpha_i\ket{i}$), and also $\braket{f_i|f_i} = 1$ (due to the unitarity of $U_R^{(1)}$), we find:
\[
p = \left|\sum_i|\alpha_i|^2\ket{f_i}\right|^2 = \sum_i|\alpha_i|^4 + 2\sum_{i>j\ge 0}|\alpha_i|^2|\alpha_j|^2\text{Re}(\braket{f_i|f_j}).
\]

When $\ket{f_i} \equiv \ket{f_j}=\ket{F},$ for all $i,j$, the above quantity attains its maximum of $p=1$. In this case, after Eve's attack, the system described by Equation~\eqref{eq:wk:walkresult} is $\sum_i\alpha_i\ket{i}\otimes\ket{F} = \ket{\psi_\sigma}\otimes\ket{F}$.  Due to the Cauchy-Schwarz inequality $\text{Re}(\braket{f_i|f_j}) \le 1$.  If, however, one or more of the $\text{Re}(\braket{f_i|f_j}) < 1$ for any of the $(\ket{f_i}, \ket{f_j})$ pairs which appear in the expression above (i.e., for those where $\alpha_i$ and $\alpha_j$ are non-zero), it is obvious that $p < 1$ and so Eve would be detected.

Therefore, to avoid detection, it must be that $\text{Re}(\braket{f_i|f_j}) = 1$ for all $i,j$ where $\alpha_i$ and $\alpha_j$ are non-zero, implying $\ket{f_i} \equiv \ket{f_j}$.  Indeed, if we write $\ket{f_j} = x\ket{f_i} + y\ket{\zeta}$, where $\braket{f_i|\zeta} = 0$, then $\text{Re}(\braket{f_i|f_j}) = 1 = \text{Re} (x)$.  Of course $|x|^2 + |y|^2 = 1$ (since $\braket{f_j|f_j} = 1$) and so:
\begin{equation*}
|x|^2+|y|^2 = 1\\
\Rightarrow \text{Re}^2 x + \text{Im}^2 x + |y|^2 = 1\\
\Rightarrow \text{Im}^2x + |y|^2 = 0.
\end{equation*}
This implies both $\text{Im}(x) = 0$ and $y=0$.  Since $\text{Re}(x) = 1$, we conclude $x=1$ and so $\ket{f_i} = \ket{f_j}$.

Since Alice could have chosen any QW in $\mathcal{Q}$, all possible $(i,j)$ pairs are covered (i.e., at least one QW in $\mathcal{Q}$ is guaranteed to produce a state where $\alpha_i$ and $\alpha_j$ are non-zero) and since Eve does not know which QW was chosen, it must be that $\ket{f_i} \equiv \ket{f_j} \equiv \ket{F}$ for all $i,j$.

Thus, after the first iteration, to avoid detection, it must be that the state of Eve's quantum memory is in the state $\ket{F}$, independently of Alice's and Bob's raw key and operations.  Thus, Eve is not able to extract any information during the first iteration. Furthermore, since she is fully aware of the state of her quantum memory in this case (i.e., she knows the state $\ket{F}$), the above arguments may be repeated inductively for the remaining iterations of the protocol, leading to the conclusion that the protocol is robust.
\end{proof}

The above proof of robustness placed certain requirements on the set of QW $\mathcal{Q}$, but can such a set even exist?  We show that, at least for all odd $P$, such a set may be easily constructed.

\begin{lemma}
If $P$ is odd, then there exists a set of QWs $\mathcal{Q}$ which satisfy the requirements of Theorem \ref{thm:sqkd-robust}.
\end{lemma}
\begin{proof}
Let $(l,s),(l',s') \in \{0,1,\cdots, P-1\}\times\{R,L\}$.  We construct a QW $U_{l,s,l',s'}$ and an initial state $\ket{l_0,s_0}$ such that $\braket{l,s|U_{l,s,l',s'}|l_0,s_0} \ne 0$ and $\braket{l',s'|U_{l,s,l',s'}|l_0,s_0} \ne 0$.

Since $P$ is odd, there exits a position index $q \in \{0,1,\cdots,P-1\}$ and a value $q_0 \in \mathbb{Z}$ such that $|q_0| < P$, $q-q_0 \equiv l \text{ (mod) } P$, and $q+q_0\equiv l' \text{ (mod) } P$.  We assume that $q_0 \ge 0$; if $q_0 < 0$ the result is symmetric by simply ``flipping'' $l$ with $l'$ (in which case $q_0$ becomes non-negative).

The shift operator $S$ for our QW is simply the usual
\[
S = \sum_{i=0}^{P-1}\ket{i-1}\bra{i}\otimes \ket{R}\bra{R} + \sum_{i=0}^{P-1}\ket{i+1}\bra{i} \otimes \ket{L}\bra{L},
\]
where all arithmetic, of course, is done modulo $P$.  Our coin operator will simply be the Hadamard coin:
\[
R_c = \frac{1}{\sqrt{2}}\left(
\begin{array}{cc}
1&1\\
1&-1
\end{array}\right).
\]
We claim the desired operator is $U_{l,s,l',s'} = \left[ ( I_p\otimes R_c)\cdot S\right]^{t+1}$. (Note that the shift operator is applied before the coin in this case to simplify the construction) Now, consider the initial state $\ket{q+1,R}$.  After the first step of the QW (i.e., after applying $(I_p\otimes R_c)\cdot S$), the QW evolves to the state $\frac{1}{\sqrt{2}}\ket{q}(\ket{R}+\ket{L})$.  It is not difficult to see that, after $t$ additional steps with this QW, \emph{but before the final application of $I_p\otimes R_c$ on the $(t+1)$-th step}, the quantum state evolves to:
\[
\alpha\ket{l,R} + \beta\ket{l',L} + \ket{\phi},
\]
where $|\alpha| \ne 0, |\beta| \ne 0$, and $\ket{\phi}$ is a non-normalized state orthogonal to both $\ket{l,R}$ and $\ket{l',L}$.  Finally, after the last $I_p\otimes R_c$, the state becomes
\begin{eqnarray*}
U_{l,s,l',s'}\ket{q+1,R} = \frac{1}{\sqrt{2}}(\alpha\ket{l,R} + \alpha\ket{l,L} \\+ \beta\ket{l',R} - \beta\ket{l',L}) + \ket{\phi'},
\end{eqnarray*}
with $\ket{\phi'}$ being a state orthogonal to $\ket{l,R}, \ket{l,L}, \ket{l',R},$ and $\ket{l',L}$, thus yielding the desired state.  Taking $\mathcal{Q} = \bigcup_{l,s,l',s'}\left\{U_{l,s,l',s'}\right\} \cup \{I\}$ proves the result.
\end{proof}

Finally, we should notice that the robustness of this SQKD protocol is independent of the existence or absence of  quantum memories. In fact, Eve's attack requires a stable quantum memory, in which she keeps her ancillary system during the execution of the protocol. On the other hand, Alice does not need any quantum memory in order to share the key with Bob at the end, and Bob is, of course, restricted to classical operations. Therefore, without a quantum memory Eve cannot even conduct the attack, whereas even if she has access to a quantum memory, she is not able to extract any useful information about the key without being detected by Alice and Bob.

\section{Practical attacks}
\label{sec:prac_att}

While our paper presents theoretical cryptographic proposals, and a detailed analysis of practical attacks is out of its scope, it is worthy presenting a short discussion of possible attacks and countermeasures for the case of optical implementations. The term practical attacks refers to attacks during which Eve is taking advantage of possible loopholes in the implementation of the protocols, i.e., the fact that the setups used for the implementation of the protocols are not perfect, can seriously compromise the security of the key. Such attacks have been thoroughly investigated in the literature and several countermeasures have been proposed in different setups and scenarios. For an overview of the recent progress and current status of this area of QKD, see the following detailed reviews~\cite{sca:bsc:csr:dus:lut:pee:09,jai:sti:kha:els:mar:leu:16,dia:lo:qi:yua:16,bed:arr:lin:17,dix:etal:17}.

One of the most studied of such attacks is the photon number splitting attack (PNS), which is based on the fact that there are no perfect single-photon sources~\cite{bra:lut:mor:san:00,lut:00,lut:jah:02}. Instead, the current sources emit in general multi-photon pulses, whose photon number statistics are described by a Poisson distribution. Eve, who is considered all powerful and bounded only by the laws of physics, can thus, by placing herself in front of Alice, detect genuine multi-photon pulses, extract one photon from each, and send the rest to Bob through a lossless channel, while blocking single-photon pulses. Due to the fact that the quantum channel connecting Alice and Bob has losses exponential in the channel length, there exists a maximal distance, known to Eve, below which Bob is not able to spot Eve's interference. By storing the extracted photons in her quantum memory, Eve can measure them in the correct basis upon the classical communication between Alice and Bob, during which they publicly reveal their choices of preparation/measurement bases. Since all the photons of the same pulse are in the same state, Eve thus has the key shared by Alice and Bob.

The standard technique used to defend against a PNS attack is by introducing the so-called decoy states~\cite{hwa:03}. In addition to the signal states, from which the key is obtained, Alice sends coherent states $\ket{e^{i\theta}|\alpha|}$, with phase $\theta$ chosen uniformly at random, and the variable intensity $I\propto |\alpha|$. Note that such decoy pulses are to Eve indistinguishable from the signal ones. Thus, Alice and Bob can subsequently detect Eve's interference (extracting single photons from multi-photon pulses) by comparing the yields of signal and decoy states (given the channel loss $\ell$, the yield $y$ is defined as $y=1-\ell$~\cite{hwa:03}). For more details, see~\cite{lo:ma:che:05}, as well as subsequent improvements and modifications~\cite{wan:05a,wan:05b,har:ett:hug:nor:05,ma:qi:zha:lo:05,wan:wan:bjo:kar:07,ros:pet:har:ric:dal:tya:mcc:nam:bae:had:hug:nor:09,luc:dyn:fro:yua:shi:15}. This method, developed for standard one-way QKD schemes, can be straightforwardly applied to our second proposal, which is a one-way protocol as well. It can also be applied to our first and third two-way proposals. Indeed, in our first proposal, as Bob does not perform any measurement, it is Alice who performs the yield estimation upon receiving back the pulses. The same can be done by Alice alone for the photons reflected by Bob in our third proposal, in which in addition the yield check could be done for the pulses measured by Bob. As mentioned above, the details of the techniques depend on particular implementations and are beyond the scope of our theoretical study.

While the PNS attack is applicable to most of the protocols that use imperfect photon sources, the above description of its particular implementation is given on the example of a standard QKD {\em one-way} scheme. Thus, it has to be re-examined when applied to different protocols. The crucial feature of the standard QKD protocol is the exchange of classical information between Alice and Bob, which allows Eve to extract the key exchanged. Therefore, since such exchange is present in our second and third protocol, the above described PNS attack is applicable to those protocols as well. Note though that in the case of the third, {\em two-way} protocol, Eve can possibly extract information only upon intercepting the pulses re-sent from Bob to Alice. Indeed, in the third protocol the key is obtained from the cases in which Bob and, upon receiving them back, Alice too, perform measurements in the computational basis, thus sharing the same set of bits. Extracting photons from the pulse before it came to bob, and consequently before his measurement, gives Eve no information about the key.

Nevertheless, our first, {\em two-way} QKD protocol, is considerably different from the standard QKD ones, as Alice and Bob reveal {\em no classical information} regarding their quantum operations (they only exchange information regarding the cases used for verification procedure, which do not contribute to the key generation). Thus, Eve's task is more difficult than in the case of standard QKD protocols. What Eve can do is to extract {\em two} photons from each pulse, one on the way from Alice to Bob, and another on the way back to Alice, and compare their states, $\ket\psi$ and $\ket{\psi(r)}$, in the attempt to learn the key $r$. Note that, even in the noiseless scenario, the described comparison does not have perfect efficiency, unlike the standard application of the PNS attack in which Eve learns the key with certainty. Moreover, in the case of our protocol, Eve can attack only three or more photon pulses, thus decreasing the efficiency of her attack with respect to the standard one, which makes use of more probable two-photon pulses as well. For example, for the commonly used order of the mean photons per pulse, $\mu = 0.2$, the probability for emitting three or more photons is $p(n\geq 3) \approx 0.001$, while the probability to emit exactly two photons (the ``deficit'' with respect to the standard PNS attack) is of the order of magnitude higher, $p(n=2) \approx 0.016$, where $n$ is the number of photons per pulse emitted. 

Finally, we would like to note that, although in practical attacks Eve is assumed to be all powerful, exceeding the current technological equipment used by everyday users, not all practical attacks are based on the same level of subtle equipment. In the case of the PNS attack, Eve should be able to perform photon non-demolition number measurements, a task beyond any current and (at least mid-term) foreseeable technology.

Nevertheless, there exist other practical protocols that do not requite such sophisticated technology. Below, we briefly analyse three such kinds of attacks, extensively studied in the literature: the Trojan horse, the detector blinding and the time-shift attacks.

The Trojan horse attack is one of the first attacks ever considered and since then it has been thoroughly investigated and continuously developed in different contexts. In a nutshell, Trojan horse attacks benefits from the imperfections in the quantum channel between Alice and Bob that allows for Eve's interference by modulating Alice's pulses, sending them to Bob and analysing the reflected/backscattered signal~\cite{vak:mak:hje:01,luc:cho:war:dyn:yua:shi:15,saj:min:jai:mak:17}. The first such attack benefitted from the detector imperfections, by collecting the light emitted upon the detection of the photons~\cite{kur:zar:may:wei:01}. To counter such attacks, introducing simple optical isolators suffice in one-way protocols, while for two-way protocols one needs to introduce additional monitoring detectors~\cite{gis:fas:kra:zbi:rib:06}.

Furthermore, we would like to briefly discuss two more attacks, namely the detector blinding and the time-shift attacks, which are both considered in the broader context of intercept and resend with faked states attacks~\cite{liz:lop:lop:16}.  In general, during an intercept and resend with faked states attack, Eve is not trying to extract information about the key from the original states that the legitimate parties exchange. Instead, she generates and sends to them classical or quantum light pulses, which are tailored in a way that she can control their measurement outcomes, while she is blocking the original states. At the end of such an attack, Eve and the legitimate parties share the same key, without Alice and Bob being able to detect Eve's interference. In both the aforementioned attacks, Eve is taking advantage of loopholes in the performance and efficiency of the detectors of the legitimate parties. 

First, we consider the detector blinding attacks to standard one-way QKD protocols~\cite{lyd:wie:wit:els:ska:mak:10,wie:lyd:wit:els:ska:mar:mak:leu:11}. Eve first intercepts the state that Alice sends to Bob and measures it in one of the two possible bases, that she randomly chooses. Then, she sends to one of Bob's detectors a bright light pulse according to her measurement outcome. Note that the intensity of the bright light pulse is just a bit above the detector's threshold. If Bob chooses to measure in the same basis as Eve, all the light will be directed to one of his detectors, due to the interference. The detector, which is now operating in the linear instead of the Geiger mode (avalanche photon diode), will click and Eve will now share the same key bit with Bob. If Bob chooses to measure in the complementary basis, the light will be divided in two components and its intensity will not be enough to trigger neither of the the detectors, therefore Bob will not get a click and this iteration will be discarded. Subsequently, Alice and Bob will keep for the key the bits for which Alice's preparation basis and Bob's measuring basis agree. During their classical communication, Eve will learn exactly which are these bits, therefore she will share the same key, while her interference remains unnoticed. 

This attack is ineffective for the case of our first, two-way, protocol, in which only Alice is performing the measurement on the pulses received back from Bob. Note that her performing the inverse QW, followed by the measurement in the computational basis, is equivalent to measuring in an {\em unknown} to Eve (ensured by our Holevo argument mentioned in Section ~3.1, and presented with details in~\cite{vla:rod:mat:pau:sou:15}), ``rotated'' basis, with respect to the computational one. Therefore, virtually all Eve attempts to perform the detector blinding attack would result in no detection events for Alice. Moreover, even the (rare) detections, being uncorrelated with the initial state sent by Alice, would not pass the verification procedure described in Section ~3.1.1, as well as the analogous checking rounds of the third, two-way semi-quantum protocol (when Bob reflects the pulses back to Alice and she performs the inverse walk and measures in the computational basis).

Regarding our second, one-way key distribution protocol, Bob's action is similar to the one in the standard protocols: he measures in one of the two publicly known bases. To counter such an attack, Bob can apply one of the known counter-measures proposed and analysed in~\cite{liu:lin:kur:ska:mak:ger:14,sti:14,ele:ozh:kur:gol:mak:15,wan:wan:qin:wei:zha:16,lee:par:woo:par:kim:han:moo:16}. Nevertheless, we would like to note again that our QW protocol is more complex than the standard ones based on few (typically four) quantum states, and thus its implementations might possibly invoke new challenges, a topic worth a separate study. 

Time-shift attacks take advantage of the different timing responses of the detectors. Assuming Eve knows the timings of during which each detector is (in)sensitive, allows her to, similarly as in the previous case of the detector blinding, enforce the particular outcomes of Bob's/Alice's measurements. Analogously as in the case of detector blinding attacks, such strategy cannot pass the two-way verification procedures and checking rounds of our first and  third protocols. In the case of our one-way protocol, one could employ similar methods to the ones proposed in~\cite{mak:ani:ska:06,lam:kur:07,zha:fun:qi:che:lo:08,qi:fun:lo:ma:07}, in order to defend against a time-shift attack. %Again, we should stress that the possible concrete implementations of our one-way quantum walk QKD protocol would require appropriate modifications of the aforementioned countermeasures.  

\section {Conclusions and future work}
In this paper we employed, for the first time to our knowledge, QWs in order to design and analyze new secure QKD protocols. 

Perhaps the most important contribution of this work is that it introduces the exciting possibility of using QWs for QKD purposes and may spur new research in both cryptography, and also in QWs. By themselves, QWs exhibit many fascinating properties which, as we've shown here, translate to interesting properties of QKD protocols. 

Besides the theoretically interesting intersection of two unique and fascinating fields of quantum information science, there are also potential practical benefits in pursuing this investigation. Recently, it has been shown in numerous studies~\cite{bec:tit:00,cer:bou:kar:gis:02,bru:chr:eke:eng:ber:kas:mac:03,nik:alb:05,she:sca:10,cha:15,nik:ran:alb:06,cha:05} that when Alice and Bob exchange higher dimensional systems (qudits instead of qubits), the respective QKD protocols (in fact both kinds of prepare-and-measure and entanglement-based) \textit{tolerate more noise} than the 2-dimensional ones. This is also confirmed by our results for the high-dimensional states generated by means of QWs, where the noise tolerance increases for higher dimensions. While such systems hold interesting theoretical properties, it could also be that, in a future quantum infrastructure, the generation of these QW states would be easier compared to other higher dimensional systems. Indeed, producing such states may not need the high entanglement of many qubits -- instead they could be generated through the evolution of a single-qubit walker on, for instance, a multi-node quantum network.

In what follows, we point out some directions of future work. First, it would be interesting to perform a more detailed study on the two verification procedures presented in Sections~\ref{sec:stdver} and~\ref{sec:bellver} and compare them with respect to various attack strategies. Moreover, one could analyze the relation between the two for concrete cases of Eve's cheating strategies in the presence of noise.

We have presented our protocols in the case that the parties use lossless channels, which is the common first step, when it comes to proving security of new theoretical QKD protocols~\cite{sca:bsc:csr:dus:lut:pee:09,lu:fun:ma:cai:11,ren:gis:kra:05,bea:luc:man:ren:13,cha:15}. A relevant future direction of work would be to conduct a detailed analysis of the effects of lossy channels on the number of secret key bits that the parties share in the end, and determine the relevant measures, such as distance vs key-rate, etc.

In Section~\ref{sec:oneway}, we proved the security of the one-way protocol, but still some improvements could be done. In particular, one could find an analytical solution for the optimal choice of QW parameters or, alternatively, given particular QW parameters, to find an analytical solution for the value of $c$ from Equation~\eqref{eq:prot1:cval}. Another interesting question would be to understand the maximally tolerated noise as the dimension of the QW $P \rightarrow \infty$ (and also $T_{\max} \rightarrow \infty$). For instance, in~\cite{cha:15}, a high-dimensional QKD protocol was introduced (not using QWs, but simpler states), which could suffer a disturbance up to $50\%$ as the dimension of the state sent by Alice approached infinity (for a novel protocol for multiparty quantum key management that implements $d$-level measurements, see~\cite{xu:che:dou:yan:li:15}). Can we construct a QW-QKD protocol with similar features?  Does our protocol approach this disturbance level for high $P$?
Moreover, studying and employing other QW models (perhaps the memory-based QWs described and analyzed in, e.g.,~\cite{get:10,get:jar:mis:14,roh:bre:gil:13,kra:15,aha:amb:kem:vaz:01,bru:car:amb:03}) or QWs on different graphs, would be interesting -- our key-rate equation would generalize to these cases, the only change would be the value of $c$; perhaps different QW models, or different graphs, would produce more optimal values, thus increasing the key-rate.

The SQKD protocol we proposed lacks of a proof of security beyond robustness. As we already mentioned in Section~\ref{sec:semi-qkdscheme}, this proof is technically very challenging due to the high-dimensional QW states and the use of a two-way channel. Hence, computing analytically the key-rate is extremely hard and moreover, the numerical simulation is equally challenging, even for low-dimensional walks. Nevertheless, we believe that obtaining the key rate is not impossible, and we expect that this analysis will yield quite high error-tolerance. A first step towards this direction would be to try to reduce this protocol to a simpler one (for instance, the one in~\cite{boy:ken:mor:07}, for which there is a security proof~\cite{krawec2015security}) and prove that it is at least as secure. This reduction does not seem to be a straightforward task and requires a thorough analysis.

Finally, concrete proposals for possible practical implementations of our protocols deserve a separate investigation, and we leave them as future work. While such practical implementations can be non-trivial, we believe that our theoretical proposals are quite promising. In particular, we are optimistic that the experimental advances in the field of the generation and manipulation of high-dimensional states mentioned above, combined with the on-going progress in the field of QW realisations, will permit the implementation of our protocols in higher dimensions. Furthermore, as we argued before, high-dimensional QW states (at least in the case of our second protocol) are not necessary in order to obtain a secure and robust protocol. Along these lines, our proposals could be seen as alternatives to the already existing QKD protocols, which can be implemented with the current technology. Perhaps, the most relevant setups for QKD purposes would be the optical implementations of QWs~\cite{bro:etal:10,san:etal:12,goyal:etal:13,goyal:etal:15,knight:etal:03,schreiber:etal:10,rol:sor:06,schreiber:etal:12,cardano:etal:15}. For a detailed and complete review of QW experimental realisations, we refer to the  book~\cite{man:wan:13}.

\section*{Acknowledgements}
WK would like to acknowledge the hospitality of SQIG--Security and Quantum Information Group in IT -- Instituto de Telecomunica\c c\~oes, in Lisbon, during his visit while working on this project.
CV acknowledges the support from DP-PMI and FCT (Portugal) through the grant PD/BD/ 52652/2014.
CV, PM, NP and AS acknowledge the support of SQIG-Security and Quantum Information Group. PM, NP and AS also acknowledge the support from UID/EEA/50008/2013. 
NP acknowledges the IT project QbigD funded by FCT PEst-OE/
EEI/LA0008/2013. 
PM and AS acknowledges the FCT project Confident PTDC/EEI-CTP/4503/2014. 
A.S. also acknowledges the support of LaSIGE Research Unit, ref. UID/CEC/00408/2013. 

\bibliographystyle{unsrt}
\bibliography{QWQKD}

\end{document}